\documentclass[11pt,a4paper]{article}
\usepackage{empheq}
\usepackage{geometry}
\usepackage{tikz}
\usepackage{amsthm}
\usepackage{amsmath}
\usepackage{amsmath,bm}
\usepackage{graphicx}
\usepackage{bbm}
\usepackage{booktabs}
\usepackage{graphicx} % Allows including images
\usepackage{amsfonts}
\usetikzlibrary{scopes,matrix,positioning}
\usepackage{algorithm,algorithmic}
\usepackage{blindtext}
\usepackage{amssymb}
\usepackage{authblk}
\usepackage{mathtools}
\usepackage{caption}
\usepackage{textcomp}
\usepackage{enumitem}
\usepackage{hyperref}
\usepackage{footnote}
\usepackage{threeparttable}
\usepackage[titletoc]{appendix}
\usepackage{authblk}
\usepackage[comma,round]{natbib}
\usepackage{setspace} 

\providecommand{\keywords}[1]{\textbf{\textit{Keywords:}} #1}

\newtheorem{assumption}{Assumption}

\newtheorem{thm}{Theorem}

\newtheorem{cor}{Corollary}
\newtheorem{lem}[thm]{Lemma}

\newcommand\scalemath[2]{\scalebox{#1}{\mbox{\ensuremath{\displaystyle #2}}}}

\makeatletter
\renewcommand\@biblabel[1]{}
\renewenvironment{thebibliography}[1]
     {\section*{\refname}%
      \@mkboth{\MakeUppercase\refname}{\MakeUppercase\refname}%
      \list{}%
           {\leftmargin0pt
            \@openbib@code
            \usecounter{enumiv}}%
      \sloppy
      \clubpenalty4000
      \@clubpenalty \clubpenalty
      \widowpenalty4000%
      \sfcode`\.\@m}
     {\def\@noitemerr
       {\@latex@warning{Empty `thebibliography' environment}}%
      \endlist}
\makeatother

\tikzset{
  mymx/.style={matrix of math nodes,nodes=myball,column sep=4em,row sep=-1ex},
  myball/.style={draw,circle,inner sep=0pt},
  mylabel/.style={midway,sloped,fill=white,inner sep=1pt,outer sep=1pt,below,
    execute at begin node={$\scriptstyle},execute at end node={$}},
  plain/.style={draw=none,fill=none},
  sel/.append style={fill=green!10},
  prevsel/.append style={fill=red!10},
  route/.style={-latex,thick},
  selroute/.style={route,blue!50!green}
}

\tikzset{   cirwhite/.style={draw=gray,circle,fill=white,minimum size=1pt,inner sep=2pt,line width=0.2mm},
            cirred/.style={draw=gray,circle,fill=red,minimum size=1pt,inner sep=2pt,line width=0.2mm},
}

\newenvironment{remark}[1][Remark]{\begin{trivlist}
\item[\hskip \labelsep {\bfseries #1}]}{\end{trivlist}}

\author[1]{Konul Mustafayeva}
\author[2]{Weining Wang}

\affil[1]{\footnotesize Department of Mathematics, King's College London, konul.mustafayeva@kcl.ac.uk}

\affil[2]{\footnotesize Department of Economics, City University London, weining.wang@city.ac.uk}

\date{\vspace{-5ex}}
 \title{Non-Parametric Estimation of Spot Covariance Matrix with High-Frequency Data}

\begin{document}
%\begin{singlespacing}
\maketitle  
%\end{singlespacing}

\begin{abstract}
Estimating spot covariance is an important issue to study, especially with the increasing availability of high-frequency financial data. We study the estimation of spot covariance using a kernel method for high-frequency data. In particular, we consider first the kernel weighted version of realized covariance estimator for the price process governed by a continuous multivariate semimartingale. Next, we extend it to the threshold kernel estimator of the spot covariances when the underlying price process is a discontinuous multivariate semimartingale with finite activity jumps. We derive the asymptotic distribution of the estimators for both fixed and shrinking bandwidth. The estimator in a setting with jumps has the same rate of convergence as the estimator for diffusion processes without jumps. A simulation study examines the finite sample properties of the estimators. In addition, we study an application of the estimator in the context of covariance forecasting. We discover that the forecasting model with our estimator outperforms a benchmark model in the literature.
\end{abstract}
\keywords{high-frequency data; kernel estimation; jump; forecasting covariance matrix}
MOS subject classification: 62F12, 62G05, 60J75.
\section{Introduction}
Spot covariance has important applications in studying the intraday patterns of the covariance process, co-jump tests  (Bibinger and Winkelmann \citeyearpar{BibingerWinkelmann}) and estimating
parametric multivariate stochastic volatility models (Kanaya and Kristensen \citeyearpar{Kanaya}). Moreover, understanding covariance dynamics is crucial for effective portfolio choice, derivative pricing, and risk management.  %Recent studies in financial econometrics and empirical finance reject the constancy of covariances over time based on strong empirical evidence for the time dependent daily and long-term correlations between the returns of asset prices. A considerable effort has been put into developing suitable mathematical models to capture covariance dynamics, see e.g. Shephard \citeyearpar{Shephard}.
The availability of high-frequency intraday data of asset returns has given rise to several approaches for estimating integrated (co)variances and spot variances. While the literature proposes few measures of integrated covariance, see e.g. Barndorff-Nielsen and Shephard \citeyearpar{barndorff2004a}, Hayashi and Yoshida \citeyearpar{HayashiYoshida},
there is sparse literature on empirical approaches and statistical theory to estimate spot covariances with high-frequency data.

%\textbf{Virtually every aspect of financial economics involves the volatility of a financial instrument over a specified period}(paraphrase). Volatility is a latent process, i.e. is not observable in the market. The best we can hope for is to measure it statistically.
%Over the last 20 years there have been the rapid advances in the theory and application of volatility measurement from high-frequency data. 
In this paper, we consider the nonparametric filtering of spot covariance with high-frequency financial data. Our study is at the intersection of two fields of literature. The first strand of literature is on estimating integrated covariance matrices over a fixed period. This topic has been studied extensively in high-frequency econometrics. For example, the highly celebrated paper by Barndorff-Nielsen and Shephard \citeyearpar{barndorff2004a} makes important contributions to the use of realized covariance to estimate integrated covariance matrix in a setup without market microstructure noise. The quasi-maximum likelihood estimator by A\"{i}t-Sahalia et al. \citeyearpar{SahaliaFanXiy}, the multivariate pre-averaging estimator by Christensen et al. \citeyearpar{ChristensenPodolskijVetter}, the two-scale estimator by Zhang \citeyearpar{Zhang} are robust to microstructutre noise.  However, all above mentioned realized covariance estimators do not account for jumps in the underlying price process.

The second strand focuses on spot volatility estimation. Several approaches of estimating spot volatility were proposed. Foster et al. \citeyearpar{FosterNelson} were the first to introduce the spot volatility estimator: rolling and sampling filters. Later, kernel-type estimators were introduced in Fan and Wang \citeyearpar{FanWang} and Kristensen \citeyearpar{Kristensen}.  These estimators of spot variance neglect the microstructure noise and jumps. The examples of spot variance estimators accounting for microstructure noise include Zu and Boswijk \citeyearpar{ZuYang}, Bos et al. \citeyearpar{Bos}, Mykland and Zhang \citeyearpar{MyklandZhang}. Yu et al. \citeyearpar{YuFang} extend kernel spot volatility estimator of Kristensen \citeyearpar{Kristensen} to the case when the underlying price process has jumps.

The estimation of spot covariance matrix is, however, an area that has been studied the least. For a multi-dimensional continuous semimartingale log-asset price process  Bibinger et al. \citeyearpar{Bibinger} propose an estimator for spot covariance which is constructed based on a local average of block-wise parametric spectral covariance estimates. Aiming to fill this gap in the literature we study the spot covariance estimation of both continuous and discontinuous semimartingales.   

Our contribution is following. First, for a setup without jumps, we study asymptotic properties of the kernel covariance estimator, which was mentioned in Kristensen \citeyearpar{Kristensen} as an extension to the multivariate case and was left for the future research. Second, we propose the threshold kernel covariance estimator when the underlying price process is a discontinuous semimartingale with finite activity jumps. We derive the asymptotic distribution of this estimator for a fixed bandwidth. The estimator is an extension to the multivariate case of the threshold kernel volatility estimator proposed by Yu et al. \citeyearpar{YuFang}. Third, we conduct numerical studies to examine finite sample properties of both estimators. Next, we study an application of the kernel estimator in the context of covariance forecasting.

In a setup without jumps the estimator is a kernel-weighted version of the standard integrated covariance estimator, which depends on a kernel function and choice of bandwidth. It can be regarded as a kernel regression in the time domain. The bandwidth choice allows us to focus on the covariance behavior at specific points in time, and give different weights to the covariance matrix over the window used. As the bandwidth shrinks to zero, the spot covariance can be extracted. We establish asymptotic normality of the estimator for both fixed and shrinking bandwidth. The proofs are component-wise. We construct our proofs referring to the techniques of Barndorff-Nielsen and  Shephard \citeyearpar{barndorff2004a} and Kristensen \citeyearpar{Kristensen}. We first derive the mean and covariance of the estimator. We then derive the asymptotic distribution by employing central limit theorem for triangular arrays and Cram\'{e}r-Wold device. We also prove the asymptotic normality for the threshold kernel estimator with fixed bandwidth. In the proof of this theorem we combine our results from the first theorem, techniques from Yu et al. \citeyearpar{YuFang} and employ Cram\'{e}r-Wold device. In simulation study we examine the finite sample properties of both estimators using the integrated mean square error and the integrated bias performance measurements.

The rate of convergence of both estimators is $\sqrt{n}$. The local method of moments estimator of the spot covariances of Bibinger et al. \citeyearpar{Bibinger} attains slower optimal rate of convergence ($\sqrt{n^4}$). However, it should be noted this is due to the fact that Bibinger et al. \citeyearpar{Bibinger} consider the setting with market microstructure noise, whereas we target for a complementary jump case. The kernel and threshold kernel covariance estimators are fairly easy to implement. 
%Second, we propose an estimator of spot covariance matrix when the underlying price process is discontinuous semimartingale with finite activity jumps. The estimator is an extension to the multivariate case of the threshold kernel volatility estimator proposed by Yu et al. \citeyearpar{YuFang}. 

In terms of applications of this kernel covariance estimator, considerable efforts has been put into covariance forecasting, see e.g. Alexander \citeyearpar{Alexander}, Andersen et al. \citeyearpar{AndersenBollerslevChristoffersen}. Multivariate GARCH models are a standard tool used in modelling and forecasting covariances. However, more recent studies propose models based on high-frequency data and options implied data. In a comprehensive empirical study by Symitsi et al. \citeyearpar{Symitsi} several approaches to the covariance forecasting are compared based on statistical and economic criteria. In this study the authors conclude that models based on high-frequency data offer a clear advantage in terms of statistical accuracy. In particular, a Vector Heterogeneous Autoregressive (VHAR) model achieves the best performance amongst the competing models. The VHAR model is a linear combination of past daily, weekly and monthly realized covariance estimators of Barndorff-Nielsen and Shephard \citeyearpar{barndorff2004a}. 

Motivated by this we use the VHAR model to forecast covariance, however instead of the realized covariance estimator we use newly proposed kernel covariance estimator. We further show that with the VHAR model the kernel covariance estimator outperforms the benchmark realized covariance estimator in all three measures of accuracy: the Euclidean loss function, the Frobenius distance and the multivariate quasi-likelihood loss function.
%Spot covariance estimators has been also directly used in financial markets. In the univariate case, the the estimator for realized spot variance was adopted by Kanaya and Kristensen \citeyearpar{Kanaya} for the estimation of parameters of stochastic volatility models. In a similar spirit, the application of the realized spot covariance estimator can be used for the parameters estimation of stochastic volatility models for multivariate data.  Moreover, spot covariance estimates are a necessary building block for co-jump tests, see e.g. Bibinger and Winkelmann \citeyearpar{BibingerWinkelmann}. Other applications of the estimator include intraday correlation risk measurement, testing co-jumps, market monitoring.

The paper is structured as follows. In Section \ref{sec:Setup_Estimators_covariance} we review theoretical setup of the problem and the kernel covariance estimator which was proposed in Kristensen \citeyearpar{Kristensen} and left for the future research. In Section \ref{sec:AsymProp} we study the asymptotic properties of the estimator for a fixed and small (tending to zero) bandwidth. In Section \ref{sec:JumpCaseEstimator} we introduce the setup with jumps, propose the estimator for jump case and derive its asymptotic distribution. In Section \ref{sec:simStudy} we conduct Monte Carlo simulations and investigate the finite sample properties of both estimators. In Section \ref{sec:Applications} we present an application of the estimator in the context of covariance forecasting. Finally, in Section \ref{sec:Conclusion} we summarise our findings.

\section{Kernel Covariance Estimation }\label{sec:Estimators_covariance}
\subsection{Theoretical Setup and the Kernel Covariance Estimator}\label{sec:Setup_Estimators_covariance}
In this section we start by considering a multidimensional continuous semimatingale, describe the theoretical setup and review the kernel covariance estimator in Kristensen \citeyearpar{Kristensen}. Our aim is to accurately estimate the spot covariance matrix of a fixed $d$-dimensional log-price process $\left(X(t)\right)_{t\geq 0}=
\left(X_{1}(t), X_{2}(t), ..., X_{d}(t)\right)_{t\geq 0}$. We assume that $X(t)$ follows a continuous semimartingale
\begin{equation}
X(t)= X(0)+ \int_0^t\mu(s)ds + \int_0^t\theta(s)dW(s), \ \ \ t\in \left[0,T\right], \label{stochasticProcess}
\end{equation}
defined on a filtered probability space $(\Omega,\mathcal{F},(\mathcal{F})_{t\geq 0},P)$, with an initial condition $X(0)\in\mathbb{R}^d$, the drift vector $\mu(t)$, the $d$-dimensional standard Brownian motion $W(t)$ and the instantaneous volatility matrix $\theta(t)$ which has elements that are all c\`{a}dl\`{a}g. The latter yields the $(d\times d)$-dimensional \textit{spot covariance matrix} $\Sigma(t) = \theta(t)\theta(t)^\top$, which is our object of interest. We also denote the \textit{integrated covariance matrix} by $\Sigma^*(t)=\int_0^t\Sigma(s)ds$. We consider the finite and fixed time horizon $[0,T]$ with $n+1$ high-frequency discrete observations $X_k(t_0),X_k(t_1),...,X_k(t_{n-1}),X_k(t_n)$ of the realization of $k$-th asset, with $k=1,2,...,d$. For an arbitrary partition $0=t_0<t_1<...<t_n=T$ of the interval $[0,T]$ we require that $\max_{i=1,...,n}\vert t_i - t_{i-1}\vert$ approaches zero under the asymptotic limit. For simplicity, we consider the case of equally spaced and synchronous observation times. We denote $\delta=T/n$, so that $t_i=i\delta$ for $i=1,2,\cdots,n$.

A kernel is a non-negative integrable function $K$ satisfying the following condition: $\int_{\mathbb{R}}K(u)du=1$. The kernel weighted measure of the integrated covariance, which is an extension of the measure of the integrated variance introduced in Kristensen \citeyearpar{Kristensen}, is of the following form
\begin{equation}
KCV(\tau)=\int_0^TK_h(s-\tau)\Sigma(s)ds, \label{integratedCovariance}
\end{equation} 
where the function $K_h(z)$ is given by $K\left(\frac{z}{h}\right)/h$, satisfies $\int_{\mathbb{R}}K(z)dz=1$, and $h>0$ is the fixed bandwidth. $KCV(\tau)$ delivers a kernel weighted average of the quadratic covariation.

An estimator of the integrated covariance in equation \eqref{integratedCovariance} is the kernel smoothed sample average of the increments, which was mentioned in Kristensen \citeyearpar{Kristensen} as an extension of the univariate case and was left for the future research:
\begin{equation}
\widehat{KCV}(\tau)=\sum_{i=1}^nK_h(t_{i-1}-\tau)\Delta X(t_{i-1})\Delta X^\top(t_{i-1}), \label{eq:KCV_estimator}
\end{equation}
where $\Delta X(t_{i-1})=X(t_i) - X(t_{i-1})$ is the $d$-dimensional vector ($d$ is fixed) of the increments of the process $X$ over time interval $[t_{i-1},t_i]$.
As demonstrated above, for a fixed $h>0$, $KCV(\tau)$ gives a weighted measure of the integrated covariance. However, as $h\to 0$, the instantaneous covariance can be recovered at any point of continuity $\tau$ of $t \mapsto \Sigma (t)$:
\begin{equation}
\Sigma (\tau)=\lim_{h\to \infty}KCV(\tau).\label{spotCovariance}
\end{equation}
To emphasize that we are working with an estimator of the instantaneous covariance at time $\tau$, we shall denote:
\begin{equation}
\widehat{\Sigma}(\tau)=\sum_{i=1}^nK_h(t_{i-1}-\tau)\Delta X(t_{i-1})\Delta X(t_{i-1})^\top  \label{eq:estimatedSpotCovariance}
\end{equation}
Note that, $\sum_{i=1}^nK_h(t_{i-1}-\tau)\Delta X(t_{i-1})\Delta X(t_{i-1})^\top$ can be regarded as the Nadarya-Watson estimator. An overview of this types of kernel can be found in Silverman \citeyearpar{Silverman}. In the univariate case, i.e. when $d=1$, we recover the spot variance estimator from Kristensen \citeyearpar{Kristensen}.

\subsection{Asymptotic Properties of the Kernel Covariance Estimator} \label{sec:AsymProp}
In this section we state the necessary assumptions and present the two out of the three main results of the paper. Our first theorem derives the asymptotic distribution of the kernel covariance estimator for the fixed bandwidth. Theorem \ref{theorem2} proves asymptotic normality of the kernel covariance estimator for a tending to zero bandwidth. Throughout our work we shall consider the following set of assumptions:
\begin{assumption}\label{assump1}
\normalfont The processes $\mu$ and $\Sigma $ are jointly independent of $W$.
\end{assumption}
This assumption holds for a widely used stochastic volatility models, such as  Heston \citeyearpar{HestonSteven}, Hull and White \citeyearpar{HullWhite}. Assumption \ref{assump1} greatly facilitates the proof by allowing us to make all arguments conditional on $\mu$ and $\Sigma$. Under Assumption \ref{assump1}, the volatility process being independent of $W$, the model falls into the case without leverage effects. However, this assumption does not appear to be strictly necessary as demonstrated in Kanaya and Kristensen \citeyearpar{Kanaya}. 
\begin{assumption} \label{assump2}
\normalfont For any sequences  $(i-1)\delta \leq s_i \leq t_i \leq i\delta$, with $i= 1, \cdots, n$ and every $k=1,\cdots,d$,  as $\delta \to 0$
\begin{equation}
\delta \sum_{i=1}^n\vert \mu_{k}^2(s_i) - \mu_{k}^2(t_i) \vert=o(1), \ \ \ \ \ \ \delta \sum_{i=1}^n \vert \Omega(s_i) - \Omega(t_i) \vert = o(1), 
\end{equation}
where $\Omega(t)=:\left\lbrace \Sigma_{kk'}(t)\Sigma_{ll'}(t) + \Sigma_{kl'}(t)\Sigma_{lk'}(t)\right\rbrace _{k,k',l,l'=1,\cdots ,d}.$
\end{assumption}
Assumptions \ref{assump2} imposes a restriction on the local behavior of the mean and covariance processes. It allows for the deterministic patterns, jumps, and nonstationarity, and is automatically satisfied when the mean and volatility processes have continuous trajectories. In particular, standard diffusion models such as Heston \citeyearpar{HestonSteven}, Hull and White \citeyearpar{HullWhite} satisfy this assumption.
\begin{assumption} \label{assump3}
\normalfont For every $k=1,\cdots,d$ and $i=1,\cdots,n$ the quantities
\begin{equation}
\delta^{-1}\int_{t_{i-1}}^{t_i}\Sigma_{kk}(s)ds \label{eq4assump3}
\end{equation}
are bounded away from 0 and infinity uniformly in $\delta$. 
\end{assumption}
Equation \eqref{eq4assump3} in Assumption \ref{assump3} essentially means that, on any bounded interval, $\Sigma_{kk}(t)$ itself is bounded away from infinity. This is the case, for example for Cox-Ingersoll-Ross (CIR) and Ornstein-Uhlenbeck (OU) processes in Cox et al. \citeyearpar{CIR}, Uhlenbeck and Ornstein \citeyearpar{OU} respectively. The above mentioned assumptions are sufficient to derive asymptotic distribution of $\widehat{KCV}(\tau)$, however in order to get the asymptotics of $\widehat{\Sigma}(\tau)$, when $h\to 0$, the general smoothness condition needs to be imposed on the covariance process. 
\begin{assumption} \label{assump4}
\normalfont The space $C^{m,\gamma}[0,T]$ for some $m\geq$ and $0<\gamma <1$ consists of functions $f:[0,T] \mapsto \mathbb{R}$ that are $m$ times differentiable with the $m$-th derivative $f^{(m)}(t)$, satisfying 
\begin{equation}
\vert f^{(m)}(t+\delta) - f^{(m)}(t) \vert \leq \mathcal{L}(t, |\delta |)|\delta|^\gamma + o(|\delta|^\gamma), \ \ \ \delta \to 0, \ \ (a.s.),
\end{equation}
where $\mathcal{L}(t, \delta)$ is Lipschitz coefficient, a slowly varying function at zero and $t \mapsto \mathcal{L}(t, 0)$ is continuous.
The mapping $t \mapsto \Sigma_{k,l}(t)$ for $k, l = (1,..., d)$ lies in $C^{m,\gamma}[0,T]$ for some $m\geq 0$ and $\gamma \geq 0$.
\end{assumption}
As stated in Yu et al. \citeyearpar{YuFang} this condition is satisfied by commonly used diffusion processes. When Assumption \ref{assump5} holds with $m=0$ and $\gamma < 0.5$ the model is driven by a Brownian motion (see e.g. Revuz and Yor \citeyearpar{RevuzYor}). 

We also impose requirements on the kernel function:
\begin{assumption} \label{assump5}
\normalfont The kernel $K:\mathbb{R} \mapsto \mathbb{R}$ 
\begin{enumerate}[label=(\alph*)]
\item satisfies $\int_{\mathbb{R}}K(x)dx=1$ and continuously differentiable, i.e. $K \in C^{1,0}$, such that 
\begin{equation}
\bar{K}_z \coloneqq \sup_{0\leq u \leq T} \vert K^{(z)}(u) \vert < \infty, \ \ \ \ z=0,1. \nonumber
\end{equation}
\item satisfies the condition that there exists some constants $\Lambda, L$ and $\Gamma_i < \infty$ such that $|K^{(i)}(u)|\leq \Lambda$, and for some $v>1$, $|K^{(i)}(u)|\leq \Gamma_i|u|^{-v}$ for $|u|\geq L$, $i=0,1$.
\item satisfies $\int_{\mathbb{R}}x^iK(x)dx =0$, $i=1,...,r-1$ and $\int_{\mathbb{R}}|x|^r|K(x)|dx , \infty$, for some $r\geq 0$. 
\end{enumerate}
\end{assumption}
The assumptions above are satisfied by most standard kernels for $r\leq 2$. When $r > 2$, $K$ is called a higher-order kernel. If $m > 2$ as well, the higher-order kernels can be used to reduce the bias in the estimation of more than twice differentiable functions. Although, as mentioned in Kristensen \citeyearpar{Kristensen}, since $m = 0$ is a usual case, Cline and Hart \citeyearpar{Cline} demonstrated that higher-order kernels can potentially reduce bias even when the object of interest is non-smooth and has jumps. 

Now we are ready to derive the asymptotics of the kernel covaraince estimator for a fixed bandwidth. 
\begin{thm} \label{theorem1}
If Assumptions \ref{assump1}-\ref{assump5} hold, we have that for fixed $h$ and any $\tau \in [0, T]$
\begin{equation}
\sqrt{\delta^{-1}}\left\lbrace \widehat{KCV}(\tau) - \int_0^TK_h(s-\tau)\Sigma(s)ds\right\rbrace \overset{\mathcal{L}}{\to} N\left(0, \int_0^TK_h^2(s-\tau)\Omega(s)ds\right), \label{asympdistr2}
\end{equation}
where $\Omega(t)$ is a $d^2\times d^2$ array with elements 
\begin{equation}
\Omega(t)=:\left\lbrace \Sigma_{kk'}(t)\Sigma_{ll'}(t) + \Sigma_{kl'}(t)\Sigma_{lk'}(t)\right\rbrace _{k,k',l,l'=1,\cdots ,d}. \label{Omega1}
\end{equation}
\end{thm}
\begin{proof}
We give the proof in several steps. First we derive the means, variances and covariances of the variates
\begin{eqnarray}\nonumber
\widehat{KCV}_{kl}(\tau)&=&\sum_{i=1}^n K_h(t_{i-1}-\tau)\Delta X_{k}(t_{i-1})\Delta X_{l}(t_{i-1})\\
&=&\sum_{i=1}^nK_h(t_{i-1}-\tau)\left(X_{k}(t_i) - X_{k}(t_{i-1})\right)\left(X_{l}(t_i)- X_{l}(t_{i-1})\right). \nonumber
\label{KCVkl2}
\end{eqnarray}
with $k,l=1,2,\cdots,d$. Second, the Theorem \ref{theorem1} is proved for the case, where the mean processes $\mu_{k}$ are identically $0$, by employing Cramer-Wold device. Finally, the latter restriction is lifted and using lemma \ref{lemma1} in Appendix \ref{proofLemma} the negligibility of non-zero drift term is shown. The proof is component-wise and based on the results and techniques employed by Barndorff-Nielsen and Shephard \citeyearpar{barndorff2004a} and Kristensen \citeyearpar{Kristensen}. See Appendix \ref{appendix:proofTheorem1} for the details of the proof.
\end{proof}

This theorem is an intermediate step in the derivation of the asymptotic distribution of the estimator for a shrinking bandwidth. The Theorem \ref{theorem1} is necessary for the proof of the asymptotic normality of the spot kernel covariance estimator in \eqref{eq:estimatedSpotCovariance}. 
%It is sometimes convenient to avoid the symmetric replication in the realized covariation matrix by employing a vech transformation. Then the limit theory can be written as follows.
%\begin{corollary}
%As $\Delta \to 0$
%\begin{equation}
%\frac{1}{\sqrt{\Delta}}\left( \normalfont\text{vech}(\widehat{KCV}(\tau)) - \normalfont\text{vech}(\int_0^TK_h(s-\tau)\Sigma(s)ds) \right) \overset{d}{\to}N()
%\end{equation}
%\end{corollary}
\begin{thm} \label{theorem2}
If Assumptions \ref{assump1}-\ref{assump5} hold with $r\geq m+\gamma$, then as $nh \to \infty$ and $nh^{2(m+\gamma)+1}\to 0$ for any $t \in (0,T)$ we have
\begin{equation}
\sqrt{\delta^{-1}h}\left\lbrace \widehat{\Sigma}(t) - \Sigma(t) \right\rbrace \overset{\mathcal{L}}{\to} N \left( 0, \Omega(t)\int_{\mathbb{R}}K^2(z)dz \right)
\end{equation}
where $\Omega(t)$ is a $d^2\times d^2$ array with elements 
\begin{equation}
\Omega(t)=: \left\lbrace \Sigma_{kk'}(t)\Sigma_{ll'}(t) + \Sigma_{kl'}(t)\Sigma_{lk'}(t)\right\rbrace _{k,k',l,l'=1,\cdots ,d}. \label{Omega2}
\end{equation}
\end{thm}
\begin{proof}
See Appendix \ref{appendix:proofTheorem2}.
\end{proof}
Bibinger et al. \citeyearpar{Bibinger} propose spot covariance estimator which is constructed based on local averages of block-wise parametric spectral covariance estimates. This is an extension of the local method of moments (LMM) in Bibinger and Reiss \citeyearpar{BibingerReiss}. Since Bibinger et al. \citeyearpar{Bibinger} consider a setting with market microstructure noise, their estimator attains the optimal rate of convergence ($\sqrt{n^4}$) which is slower compared to the convergence rate of the kernel covariance estimator ($\sqrt{n}$). The kernel estimator in equation \eqref{eq:estimatedSpotCovariance} is fairly easy to implement. 
%\subsection{Example: the bivariate case}
%paraphrase the whole section
\begin{remark}[Remark: the bivariate case.]
It is helpful to focus on the bivariate case in order to gain further understanding. We will look at the results for the assets $k$ and $l$, whose log-prices will be written as $X_{k}$ and $X_{l}$ respectively. Then the high-frequency returns at time $t_i$ is
\begin{equation*}
\Delta X_k(t_i)= X_k(t_i) - X_k(t_i-1) \ \ \ \text{and} \ \ \ \Delta X_l(t_i)= X_l(t_i) - X_l(t_i-1) \ \ \ \text{for} \  i=1,\cdots,n.
\end{equation*}
In order to avoid the symmetric replication in the covariation matrix we employ a half-vectorization, or alternatively, a vech transformation. The half-vectorization of a symmetric matrix is obtained by vectorizing only the lower triangular part of the matrix (see Kollo and Rosen \citeyearpar{KolloRosen},  L{\"u}tkeohl \citeyearpar{Lutkeohl}). In this case Theorem \ref{theorem1} tells us that joint asymptotic distribution for identifying elements of realized covariation of two assets $X_{k}$ and $X_l$ becomes
\begin{eqnarray}
\sqrt{\delta^{-1}} \left( \begin{array}{c}
\sum_{i=1}^n K_h(t_{i-1}-\tau)\Delta X_k^2(t_i) - \int_0^TK_h(s-\tau)\Sigma_{kk}(s)ds\\
\sum_{i=1}^n K_h(t_{i-1}-\tau)\Delta X_k(t_i)\Delta X_l(t_i) - \int_0^TK_h(s-\tau)\Sigma_{kl}(s)ds  \\
\sum_{i=1}^n K_h(t_{i-1}-\tau)\Delta X_l^2(t_i) - \int_0^TK_h(s-\tau)\Sigma_{ll}(s)ds \end{array} \right)\overset{\mathcal{L}}{\to}& \nonumber \\
N\left[0, \displaystyle\int_{0}^{T}K_h^2(s-\tau)\left(  \scalemath{0.8}{\begin{array}{ccc}
2\Sigma_{kk}^2(s) & 2\Sigma_{kk}(s)\Sigma_{kl}(s) & 2\Sigma_{kl}^2(s)) \\
2\Sigma_{kk}(s)\Sigma_{kl}(s) & \Sigma_{kk}(s)\Sigma_{ll}(s)+\Sigma_{kl}^2(s) 2\Sigma_{ll}(s)\Sigma_{kl}(s) & 2\Sigma_{ll}(s)\Sigma_{kl}(s)\\
2\Sigma_{kl}^2(s) & 2\Sigma_{ll}(s)\Sigma_{kl}(s) & 2\Sigma_{ll}^2(s)) \end{array}} \right)du\right]. \nonumber
\end{eqnarray}
%Here, for convenience we employed a vech transformation, in order to avoid the symmetric replication in the covariation matrix.
\end{remark}
\section{Jump Case: Threshold Kernel Covariance Estimation} \label{sec:JumpCaseEstimator}
In this section we assume that the price process is governed by a discontinuous semimartingale with finite activity jumps. We propose a threshold kernel spot covariance estimator, which is an extension of the threshold kernel spot volatility estimator in Yu et al. \citeyearpar{YuFang} to the multivariate case. Theorem \ref{theorem3} derives the asymptotic distribution of the threshold kernel covariance estimator for a fixed bandwidth.

Consider a filtered probability space $(\Omega, (\mathcal{F})_{t\in[0,T]},\mathcal{F}, P)$. Let the d-dimensional (with fixed $d$) log-price $X(t) = (X_1(t), X_2(t), ..., X_d(t))$ be defined on the this space and satisfy the following stochastic differential equation:
\begin{equation}
dX(t) = \mu(t)dt + \theta(t)dW(t) + dJ(t), \ \ \ t \in [0,T].
\end{equation}
where $\mu(t)$ is the drift vector, $\theta(t)$ is the instantaneous volatility matrix, $W(t)$ is the $d$-dimensional standard Brownian motion and $J(t)=(J_1(t),..., J_d(t))$ is a compound Poisson process with finite activity of jumps, which can be written as $J(t)=\sum_{i=1}^{N(t)}(Z_{1}(\tau_i), ..., Z_{d}(\tau_i))$
$=\left( \sum_{i=1}^{N(t)}Z_{1}(\tau_i), ..., \sum_{i=1}^{N(t)}Z_{d}(\tau_i) \right)$. Here \noindent $(N(t))_{t\geq 0}$ is a homogeneous Poisson process with constant intensity $\lambda > 0$ and $(Z_k)_{k \in \mathbb{N}}$ is a sequence of i.i.d. random variables with values in $\mathbb{R}^d$, which denotes the jump size at the jump location $\tau_i$. We assume $Z_k(\tau_i)$ for $k=1,2,...d$ are i.i.d. and independent of $N_t$. Denote the $(d\times d)$-dimensional \textit{spot covariance matrix} by $\Sigma(t) = \Theta(t)\Theta(t)^\top$.

Suppose that on a finite and fixed time horizon $[0,T]$, we have $n+1$ high-frequency discrete observations $X_k(t_0),X_k(t_1),...,X_k(t_{n-1}),X_k(t_n)$ of the realization of $k$-th asset, with $k=1,2,...,d$. Here, $0=t_0<t_1<...<t_n=T$ is an arbitrary partition of the interval $[0,T]$. Although the observations are not necessarily equidistant, we require that $\max_{i=1,...,n}\vert t_i - t_{i-1}\vert$ approaches zero under the asymptotic limit. We consider the case of equally spaced and synchronous observation times, though this assumption can easily be lifted. Denote $\delta=T/n$, so that $t_i=i\delta$ for $i=1,2,\cdots,n$. 

The quantity of interest is the spot covariance matrix $\Sigma(t)$. The threshold kernel covariance estimator, denoted by $\widehat{TCV}$, is defined as 
\begin{equation}
\widehat{TCV}(\tau) = \sum_{i=1}^nK_h(t_{i-1}-\tau)\Delta X(t_{i-1})\Delta X^\top(t_{i-1})\mathbbm{1}_{\{ \lVert  \Delta X_{t_{i -1}} \rVert  \leq dr(\delta)\}}, 
\end{equation}\label{eq:TKC}
where $\mathbbm{1}(\cdot)$ is the indicator function and $\Delta X(t_{i-1})=X(t_i) - X(t_{i-1})$ is the $d$-dimensional vector of increments of process $X$ over time interval $[t_{i-1},t_i]$. The function $K_h(x)$ is given by $K(x/h)/h$, where $h$ is bandwidth and the kernel function $K(x)$ satisfies $\int_{\mathbb{R}}K(x)dx=1$. The threshold function $r(\delta)$ is a deterministic function of the step length $\delta$. As the bandwidth $h \to 0$ we recover the spot covariance. The threshold function $r(\delta)$ has to vanish more slowly than the modulus of the continuity of the Brownian motion in order to have the convergence in probability. Thus we have the following additional assumption.
\begin{assumption}\label{assump6}
$r(\delta)$ is a deterministic function of the step length $\delta$ such that $\underset{\delta\to 0}{\lim}r(\delta)=0$ and $\underset{\delta\to 0}{\lim}\frac{\delta\log\frac{1}{\delta}}{r(\delta)}=0$.
\end{assumption}
We now can derive the asymptotics of the threshold kernel covariance estimator.
\begin{thm} \label{theorem3}
If Assumptions \ref{assump1}-\ref{assump6} hold, we have that for fixed $h$ and any $t \in [0, T]$
\begin{equation}
\sqrt{\delta^{-1}}\left\lbrace \widehat{TCV}(\tau) - \int_0^TK_h(s-\tau)\Sigma(s)ds\right\rbrace \overset{\mathcal{L}}{\to} N\left(0, \int_0^TK_h^2(s-\tau)\Omega(s)ds\right), \label{asympdistr1}
\end{equation}
where $\Omega(t)$ is a $d^2\times d^2$ array with elements 
\begin{equation}
\Omega(t)=:\left\lbrace \Sigma_{kk'}(t)\Sigma_{ll'}(t) + \Sigma_{kl'}(t)\Sigma_{lk'}(t)\right\rbrace _{k,k',l,l'=1,\cdots ,d}. \label{Omega}
\end{equation}
\end{thm}
\begin{proof}
See Appendix \ref{appendix:proofTheorem3} 
\end{proof}
The threshold kernel covariance estimator in equation \eqref{eq:TKC} is an extension of the threshold kernel estimator of the time-dependent spot volatility in Yu et al. \citeyearpar{YuFang} to the multivariate case. In Theorem \ref{theorem3} we derive asymptotic distribution for the estimator for a fixed bandwidth $h$ of the kernel. The similar results as in Theorem \ref{theorem3} was achieved for univariate case in Yu et al. \citeyearpar{YuFang}.

\section{Simulation Study}\label{sec:simStudy}
%Using Theorem \ref{theorem1} we will to show the consistency of the estimator show for the shrinking bandwidth: $\widehat{\Sigma}(\tau) \overset{P}{\to}\Sigma(\tau)$, as well as derive the asymptotic normality of the estimator $\widehat{\Sigma}(\tau)$. Furthermore, we will assume market microstructure noise and show the robustness of the estimator to the noise.
In this section we examine the performance of the kernel and threshold kernel covariance estimators. In particular, we investigate the finite-sample performances of the estimators relative to the time distance between observations. Throughout we work with bivariate stochastic volatility model. First, we examine the kernel covariance estimator in a setup without jumps and assume that asset prices, $Y(t) = (Y_1(t), Y_2(t))$, follows Heston model:
\begin{equation}
dY(t) = \mu Y(t)dt + \theta(t)Y(t)dW(t), \ \ \ \ \ \ \ \Sigma(t) = \theta(t)\theta'(t),
\end{equation} \label{eq:HestonNoJump}
where 
\begin{equation}
\Sigma(t) = \begin{pmatrix} 
\Sigma_{11}(t) & \Sigma_{12}(t)\\
\Sigma_{12}(t) & \Sigma_{22}(t)
\end{pmatrix} = \begin{pmatrix} 
\sigma_{1}^2(t) & \sigma_{1,2}(t)\\
\sigma_{1,2}(t) & \sigma_{2}^2(t)
\end{pmatrix}
\end{equation}
with the covariance $\sigma_{1,2}(t) = \sigma_{1}(t)\sigma_2(t)\rho(t)$, the drift vector $\mu(t) = (\mu_1(t),\mu_2(t))$ and a standard two dimensional Brownian motion $W(t) = (W_1(t), W_2(t))$ such that $d\left\langle W_1, W_2 \right\rangle_t = \rho dt$. The variance processes, $\sigma_i(t)$ for $i = 1,2$, follow the CIR model Cox et al. \citeyearpar{CIR}:
\begin{equation}
d\sigma^2_i(t) = \kappa_i(\theta_i - \sigma^2_i(t))dt + \eta_i\sigma_i(t)dZ_i(t).
\end{equation}
We set the correlation between asset and its volatility process to zero in order for Assumption \ref{assump1} to hold. The remaining data generating parameters are chosen to match the estimated parameter values in Barndorff-Nielsen and Shephard \citeyearpar{BandrofShephardRealizedVariance}. In our simulation we set $T=2$ (48 hours). We consider frequencies $\Delta^{-1} = 12\times 60 \times 24, 2 \times 60 \times 24, 60 \times 24, 12 \times 24, 6\times 24$ corresponding to sampling every 5 seconds, 20 seconds, 1 minute, 5 minutes and 10 minutes. In order to simulate the data using model \eqref{eq:HestonNoJump} we employ the Euler discretization scheme from Kloeden and Platen \citeyearpar{KloedenPlaten}. We simulate one trajectory of each $\{\sigma_i^2(t)\}$ for $i=1,2$ and keep them fixed. Then we run 500 Monte Carlo repetitions for  prices of two assets $\{Y_1(t),Y_2(t)\}$. In each repetition we compute $\hat{\Sigma}_{kl}(t)$ for $i=1,2$ based on sampling frequencies.

Three different estimators of instantaneous covariance: Gaussian kernel estimator, one-sided kernel estimator and beta kernel estimator are implemented. For all three estimators cross-validation was used to select the bandwidth (see Kristensen \citeyearpar{Kristensen}). We used the following integrated squared error (ISE) as the goodness-of-fit criterion:
\begin{equation}
\text{ISE}(h) = \int_{t_l}^{t_u}\left(\Sigma_{kl}(s) - \widehat{\Sigma}_{kl}(s)\right)^2ds, \ \ \ \ \ \ \ \text{for} \ \ 0<t_l<t_u<T,
\end{equation}
where $\Sigma_{kl}(s)$ and $\widehat{\Sigma}_{kl}(s)$ for $k,l = 1,2$ are the true and the estimated spot covariances. Two performance measurements are used to evaluate the finite-sample properties of the estimators: the integrated mean squared error and the integrated bias
\begin{table}
\centering
\begin{threeparttable}[b]
\captionof{table}{Interior performance of the $KCV$ estimator}\label{tab:interperf}
\begin{tabular}{lllcrrr}
\hline
 \multicolumn{3}{r}{Gaussian kernel} & \multicolumn{2}{r}{\textbf{$\text{One-sided kernel}^{*}$}} & \multicolumn{2}{r}{Beta kernel} \\
\cline{2-7}
Data Frequency    & IMSE & ISB &  \textbf{IMSE} & \textbf{ISB} &  IMSE & ISB  \\
\hline
5 seconds      & 0.14   &  0.37 & \textbf{0.11} &   \textbf{0.21}  & 0.13 &    0.28\\
20 seconds     & 0.73 &   0.63  & \textbf{0.43}  &  \textbf{0.49} & 0.66&   0.46 \\
1 minute    &  0.80   &  0.74 & \textbf{0.59} &  \textbf{0.71} & 0.76 &  0.69 \\
5 minutes       & 1.85       & 1.97   & \textbf{1.17}        & \textbf{1.24} & 2.03 &    1.43 \\
10 minutes     &  3.88 & 4.21 & \textbf{2.16} &    \textbf{2.14} & 2.85 &  3.16 \\
\hline
\end{tabular}
\begin{tablenotes}\footnotesize
\item Note: Integrated mean squared error $(\times 10^{-5})$ and integrated squared bias $(\times 10^{-5})$.
\end{tablenotes}
\end{threeparttable}
\end{table}
\begin{equation}
\text{IMSE} = \int_{t_l}^{t_u} \left[ (\Sigma_{kl}(s) - \widehat{\Sigma}_{kl}(s))^2\right]ds, \ \ \text{ISB}=\int_{t_l}^{t_u}\left[ E(\Sigma_{kl}(s) - \widehat{\Sigma}_{kl}(s))^2ds \right], \label{eq:IMSE_ISB}
\end{equation}
where $0 \leq t_l \leq t_u \leq T$. The results for the performance of the estimator of the covariance, $\widehat{\Sigma}_{12}(t)$, are reported in Table \ref{tab:interperf}. Figure \ref{fig:KC_QQ} displays QQ plot for observed standardized error terms of Kernel Covariance Estimator using minute-by-minute data.
\begin{figure}[htb]
\begin{center}
\includegraphics[width=0.7\textwidth]{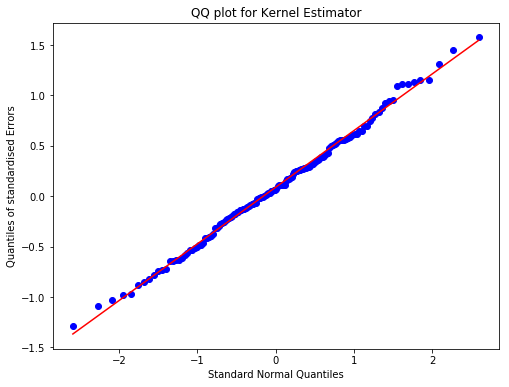}
\caption{QQ plot for observed standardized error terms of Kernel Covariance Estimator using minute-by-minute data}
\label{fig:KC_QQ}
\end{center}
\end{figure}

Next, we examine the finite sample performance of the threshold covariance estimator. Though several models combining jumps and stochastic volatility
appeared in the literature, we use the model from Bates \citeyearpar{Bates}, one of the most popular examples of the class, an independent jump component is added to the Heston
stochastic volatility model: 
\begin{equation}
dX(t) = \mu dt + \theta(t)dW(t) + dJ(t), \ \ \ \ \ \ \ \Sigma(t) = \theta(t)\theta'(t),
\end{equation} \label{eq:HestonJump}
with 
\begin{equation}
\Sigma(t) = \begin{pmatrix} 
\Sigma_{11}(t) & \Sigma_{12}(t)\\
\Sigma_{12}(t) & \Sigma_{22}(t)
\end{pmatrix} = \begin{pmatrix} 
\sigma_{1}^2(t) & \sigma_{1,2}(t)\\
\sigma_{1,2}(t) & \sigma_{2}^2(t),
\end{pmatrix}
\end{equation}
where $X(t) = (X_1(t),X_2(t))$ is log of asset prices, $\sigma_{1,2}(t) = \sigma_{1}(t)\sigma_2(t)\rho(t)$, $\mu = (\mu_1,\mu_2)$ is the drift vector, $J(t)=\sum_{i=1}^{N(t)}(Z_{1}(\tau_i),Z_{2}(\tau_i))$ is a two dimensional compound Poisson jump process and $W(t) = (W_1(t), W_2(t))$ is a standard two dimensional Brownian motion such that $d\left\langle W_1, W_2 \right\rangle_t = \rho dt$. The variance processes, $\sigma_i(t)$ for $i = 1,2$, follow the CIR model:
\begin{table}
\centering
\begin{threeparttable}[b]
\captionof{table}{Interior performance of the TKCV estimator}\label{tab:interperfTCV}
\begin{tabular}{lllcrrr}
\hline
\multicolumn{3}{r}{Gaussian kernel} & \multicolumn{2}{r}{\textbf{$\text{One-sided kernel}^*$}} & \multicolumn{2}{r}{Beta kernel} \\
\cline{2-7}
Data Frequency    & IMSE & ISB &  \textbf{IMSE} & \textbf{ISB} &  IMSE & ISB  \\
\hline
5 seconds      &  1.76  & 1.38 & \textbf{1.25} & \textbf{1.22}    & 2.34 &  1.75 \\
20 seconds     & 2.24 & 1.13   & \textbf{1.87}  & \textbf{1.34}  &2.13 & 2.03   \\
1 minute    &  3.76   & 1.45  & \textbf{2.31} &  \textbf{1.67}  & 3.54 &  2.43 \\
5 minutes       &  9.35   &  1.67  & \textbf{7.31}  & \textbf{1.35} & 3.52 & 6.67    \\
10 minutes     &  5.53 & 1.25 & \textbf{3.65} &  \textbf{7.38}  & 1.83 &  4.39 \\
\hline
\end{tabular}
\begin{tablenotes}\footnotesize
\item Note: Integrated mean squared error $(\times 10^{-5})$ and integrated squared bias $(\times 10^{-5})$.
\end{tablenotes}
\end{threeparttable}
\end{table}
\begin{equation}
d\sigma^2_i(t) = \kappa_i(\theta_i - \sigma^2_i(t))dt + \eta_i\sigma^2_i(t)dZ_i(t).
\end{equation}
As in simulations for Heston model without jumps we set $T=2$ (48 hours) and consider sampling frequencies 5 seconds, 30 seconds, 1 minute. We employ Euler discretization scheme from Kloeden and Platen \citeyearpar{KloedenPlaten} for the simulation. We simulate one trajectory of each $\{\sigma_i^2(t)\}$ and $\{J_i(t)\}$ for $i=1,2$ and keep them fixed. Then we run 500 repetitions of $(X_1(t),X_2(t))$. For each simulated path of the bivariate log asset price we compute $\widehat{TCV}$ based on sampling frequencies.
\begin{figure}[htb]
\begin{center}
\includegraphics[width=0.7\textwidth]{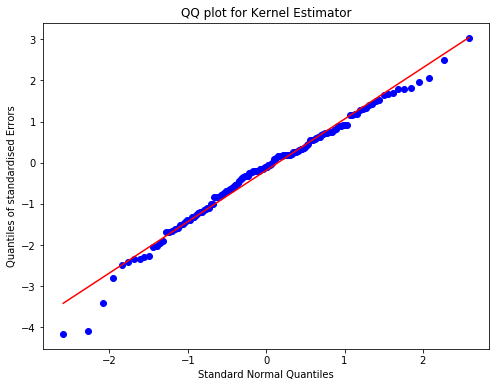}
\caption{QQ plot for observed standardized error terms of Threshold Kernel Covariance Estimator using minute-by-minute data}
\label{fig:TKC_QQ}
\end{center}
\end{figure}
We use two IMSE and ISB performance measurements in equation \eqref{eq:IMSE_ISB} for three different estimators: Gaussian, beta and one-sided kernel estimator. The results for the performance of the $\widehat{TCV}$ estimator are reported in Table \ref{tab:interperfTCV}. Figure \ref{fig:TKC_QQ} displays QQ plot for observed standardized error terms of Threshold Kernel Covariance Estimator using minute-by-minute data.

\section{Applications: Covariance Forecasting} \label{sec:Applications}
Forecasting covariance has an important economic value in the context of asset pricing and portfolio allocation. Multivariate GARCH model is a standard tool of modelling and forecasting covariances. However, the more recent approaches advocate the use of high-frequency data. 

Symitsi et al. \citeyearpar{Symitsi} undertake a comprehensive empirical comparison of two generic families of covariance forecasting models: multivariate GARCH models that employ daily data and models that use high-frequency and options data. The authors conclude that models based on high-frequency data offer both a clear advantage in terms of statistical accuracy and yield more theoretically consistent predictions leading to superior out-of-sample portfolio performance. In particular, a Vector Heterogeneous Autoregressive Model (VHAR) achieves the best performance out of the models under consideration. Motivated by this, we use the VHAR model to forecast the integrated covariance, however, when implementing for a finite sample, we use the kernel covariance estimator \eqref{eq:KCV_estimator} in Section \ref{sec:Estimators_covariance} instead of the realized covariance estimator of Barndorff-Nielsen and Shephar \citeyearpar{barndorff2004a}. 

Heterogeneous Autoregressive model (HAR), see Corsi \citeyearpar{Corsi}, was proposed as a simple way to approximate the long-memory behaviour of volatility. Vector HAR, implemented in Chiriac \citeyearpar{Chiriac}, is a multivariate extension of HAR. In the VHAR the realized covariance is expressed as a linear combination of past daily, weekly and monthly realized covariances:
\begin{equation}
RC_{t+1} = \alpha +\beta_dRC_t +\beta_wRC_{t-5:t} + \beta_mRC_{t-22:t} +\epsilon_{t+1}, \label{eq:VHAR}
\end{equation}
where $RC_t$ is obtained from Cholesky decomposition of realized covariance matrix. If $H_t$ is a matrix of realized covariances, its Cholesky decomposition gives $H_t = C_tC_t'$ and then $RC_t=vech(C_t)$. In order to allow
direct comparison among quantities defined over various time horizons, these multiperiod factors are normalized sums of the daily realized factors, i.e.
\begin{equation}
    RC_{t-k:t} = \frac{1}{k}\sum_{i=0}^{k-1}RC_{t-i}
\end{equation}
is the past $k$ day values of $RC$, $\alpha$ is a constant term and $\beta_d, \beta_w, \beta_m$ are, respectively, the parameters of daily, weekly and monthly components of the model. The covariance forecasts, $H_t$, are obtained by the reverse transformations of the $RC_t$'s. Modelling the Cholesky factors rather than covariances directly is done in order to avoid unnecessary restrictions that ensure positive definiteness.  

We simulate the log-prices of two assets and their volatilises using model \eqref{eq:HestonNoJump} in Section \ref{sec:simStudy}. Since we use simulated data, we have the true integrated covariance matrix and we propose to forecast the true covariance matrix using two measures of integrated covariance: standard in the literature realized covariance estimator of Barndorff-Nielsen and Shephard \citeyearpar{BandrofShephardRealizedVariance} and newly proposed kernel filtering of the covariance in equation \eqref{eq:KCV_estimator}. Thus we have two models for forecasting integrated covariance.  First model is VHAR model where we use the realized covariance as a measure of integrated covariance:  
\begin{eqnarray}
IC_{t+1} = \alpha +\beta_d RC_t +\beta_w RC_{t-5:t} + \beta_m RC_{t-22:t} +\epsilon_{t+1}, \label{eq:VHAR-Int}
\end{eqnarray}
where $IC$ is the half-vectorized Cholesky decomposition of the integrated covariance matrix.
\begin{table}
\centering
\resizebox{\columnwidth}{!}{
\begin{threeparttable}[b] 
\captionof{table}{The table reports the out of sample forecast loses for the 1-, 5-, 22-day horizons, respectively. The model with the lowes out-of-sample loss is market with asterisk (*).}\label{tab:VHAR}
\begin{tabular}{cccccccccc}
\hline
\multicolumn{1}{c}{} & \multicolumn{3}{c}{1-day horizon} & \multicolumn{1}{c}{}& \multicolumn{2}{c}{1-week horizon} & \multicolumn{1}{c}{} & \multicolumn{2}{c}{2-week horizon} \\
\cline{3-4} \cline{6-7} \cline{9-10}
          & & $VHAR$ & $VHAR-KCV^*$ & & $VHAR$ & $VHAR-KCV^*$ &  & $VHAR$ & $VHAR-KCV^*$   \\
\hline
$\alpha$  & & 0.3243 & 0.3213 & & 0.4987 & 0.4896 & & 0.4124 & 0.4126 \\
$\beta_d$ & & 0.6904 & 0.6064 & & 0.2443 & 0.2032 & & 0.2295 & 0.2175 \\
$\beta_w$ & & 0.6909 & 0.6028 & & 0.1765 & 0.1483 & & 0.2257 & 0.1591 \\
$\beta_m$ & & 0.8922 & 0.8374 & & 0.9007 & 0.7289 & & 0.5219 & 0.4328\\
&& & && & && & \\
$\mathcal{L}_E$ & & 0.1267 & 0.0529 & & 0.1831 & 0.0772 & & 0.2412 & 0.1841  \\
$\mathcal{L}_F$ & & 0.1387 & 0.0546 & & 0.1796 & 0.0797 & &  0.2981 & 0.1902 \\
$\mathcal{L}_Q$ & & -10.143 & -14.0537 & & -9.893 & -13.2624 & & -7.8503& -11.5561  \\ 
%&& & && & && & \\
\hline
\end{tabular}
%\begin{tablenotes}\footnotesize
%\item Note: 
%\end{tablenotes}
\end{threeparttable}
}
\end{table}
In light of this it is natural to define the VHAR-KCV model, in which we borrow the VHAR model above to predict the integrated covariance matrix, however we use kernel covariance estimator:
\begin{eqnarray}
IC_{t+1} = \alpha +\beta_d\widehat{KCV}_t +\beta_w\widehat{KCV}_{t-5:t} + \beta_m\widehat{KCV}_{t-22:t} +\epsilon_{t+1}, \label{eq:VHAR-KCV}
\end{eqnarray}
 where $\widehat{KCV}$ is the half-vectorized Cholesky decomposition of the  kernel covariance estimator in \eqref{eq:KCV_estimator}. We benchmark the VHAR-KCV against the VHAR.

In line with Symitsi et al. \citeyearpar{Symitsi} we evaluate forecasting ability of the the VHAR-KCV model \eqref{eq:VHAR-KCV} based on three multivariate loss functions and compare its performance to the performance of the benchmark VHAR model \eqref{eq:VHAR-Int}. We use the Euclidean loss function, $\mathcal{L}_E$, which is equally-weighted elements of the forecast error matrix; the Frobenius distance, $\mathcal{L}_F$, which is the extension of the mean squared error to the multivariate space and the multivariate quasi-likelihood loss function, $\mathcal{L}_Q$, which is scale invariant:
\begin{eqnarray}
&& \mathcal{L}_E = vech(\Sigma_t - H_t)'vech(\Sigma_t - H_t), \\
&& \mathcal{L}_F = Tr[(\Sigma_t - H_t)'(\Sigma_t - H_t)], \\
&& \mathcal{L}_Q = \log|H_t| +Tr(H_t^{-1}\Sigma_t). 
\end{eqnarray}
Here $Tr$ denotes the trace of square matrix, $\Sigma_t$ denotes the integrated covariance matrix at time $t$ and $H_t$ is time $t$ matrix of conditional covariance forcasts.

Results are reported in Tables \ref{tab:VHAR}. It is clear that for all forecasting horizons, the VHAR-KCV model outperforms the VHAR model which was shown to be the best model for forecasting covariance matrix in large study by Symitsi et al. \citeyearpar{Symitsi}.

\section{Concluding Remarks} \label{sec:Conclusion}
Inspired by the kernel filtering of spot volatility, in this paper we develop estimators of spot covariances for two types of the underlying price process: continuous and discontinuous semimartingales. We show the asymptotic normality of the estimators. An important result is that we are able to attain the convergence rate for both estimators, which is $\sqrt{n}$. The convergence rate of spot covariance matrix estimator for continuous martingales in a setup with microstructure noise proposed by Bibinger et al. \citeyearpar{Bibinger} is, in turn,  $\sqrt{n^{4}}$. In financially realistic scenarios, we conduct Monte Carlo experiments to study the finite sample properties of our estimators. In addition, we investigate one of the possible applications of the estimator, the forecasting of covariance matrix. We conclude that our estimator performs better in the context of forecasting than the benchmark realized covariance estimator of Barndorff-Nielsen and Shephard \citeyearpar{barndorff2004a}. One of the possible extensions of the estimators is to consider a market-microstructure noise.

%\bibliographystyle{plainnat}
%\bibliography{sample}

\pagebreak
\clearpage

\begin{appendices} 
\section{Proof of Theorem \ref{theorem1}} \label{appendix:proofTheorem1}
\subsection{Notation}\label{notation}
In a similar way to Barndorff-Nielsen and Shephard \citeyearpar{barndorff2004a} for the purpose of simplifying the proof we will use the index (or equivalently, tensor) notation instead of vector or matrix notation. We rewrite the $d$ stochastic processes $X_{k}$ $(k=1,\cdots, d)$ in equation \eqref{stochasticProcess} in index notation as
\begin{equation}
X_{k}(t)=\int_0^t\mu_{k}(s)ds + \int_0^t\theta_{k}^a(u)dW_{a}(s), \label{stochasticProcessTensor}
\end{equation}
with initial condition $X_{k}(0)=0$. Here
\begin{equation}
\Theta(t)=\{ \theta_{(k)}^a(t) \}_{k,a=1,2,\cdots,d}. \nonumber
\end{equation}
In the index notation the Einstein summation convention is used, which means if an index variable appears twice in a single expression then it implies  summation over that index. Thus \eqref{stochasticProcessTensor} is understood to mean 
\begin{equation}
X_{k}(t)=\int_0^t\mu_{(k)}(s)ds + \sum_{a=1}^d\int_0^t\theta_{(k)}^a(s)dW_{a}(s). 
\end{equation}
We apply summation convention to indices $a,b,c,d$, but not to indices $k,l,k',l'$, unless otherwise specified. 
Furthermore, we write
\begin{equation}
\theta_{kl}^{ab}=\theta_{k}^a\theta_{l}^{b}, \label{noEnsteinSummation}
\end{equation}
with similar notation for other index combination. In \eqref{noEnsteinSummation} no superscripts or subscripts are repeated and so no summation operator is generated. Combining the Einstein summation convention and the notional rule for $\theta_{kl}^{ab}$, the $(k,l)$th element of the spot covalatility matrix of model \eqref{stochasticProcess}  is 
\begin{equation}\label{eq:spotcovkl}
\Sigma_{kl}(t)=\theta_{kl}^{aa}=\sum_{a=1}^d\theta_{k}^a(t)\theta_{l}^a(t).
\end{equation}

%For convenience we also define the following
%\begin{equation}
%\mu_{ki} = \mu_k (t_i) - \mu_k (t_{i-1}),
%\end{equation}
%and
%\begin{eqnarray}
%\Gamma_{kli} = \int_{t_{i-1}}^{t_i} \theta_{kl}^{aa}(s)ds \\
%\Gamma_{kl}(T) = \int_0^T \theta_{kl}^{aa}(s)ds
%\end{eqnarray}
%Then, in fact,
%\begin{equation}
%\Gamma_{kli} = \int_{t_{i-1}}^{t_i}\Sigma_{kl}(s)ds
%\end{equation}
%\begin{equation}
%\Gamma_{kl}(T) = \int_0^T \Sigma_{kl}(s)ds
%\end{equation}
\subsection{Mean and variances}
The proof of Theorem \ref{theorem1} consists of several steps. First step is to derive the means and covariances of the variates 
\begin{eqnarray}
\widehat{KCV}_{kl}(\tau)&=&\sum_{i=1}^n K_h(t_{i-1}-\tau)\Delta X_{k}(t_{i-1})\Delta X_{l}(t_{i-1})\\
&=&\sum_{i=1}^nK_h(t_{i-1}-\tau)\left(X_{k}(t_i) - X_{k}(t_{i-1})\right)\left(X_{l}(t_i)- X_{l}(t_{i-1})\right),
\label{KCVkl1}
\end{eqnarray}
with $k,l=1,2,\cdots,d$.
Next, the Theorem \ref{theorem1} is proved for the case, where the mean processes $\mu_{k}$ $(k=1,\cdots,d)$ are identically $0$. Finally, the latter restriction is lifted. The proof is component-wise and based on the results and techniques employed by Barndorff-Nielsen and Shephard \citeyearpar{barndorff2004a} and Kristensen \citeyearpar{Kristensen}.\\
We start by computing the expectation of $\widehat{KCV}_{kl}(\tau)$ in equation \eqref{KCVkl1}.
\begin{eqnarray} \nonumber
\mathrm{E}\left[ \widehat{KCV}_{kl}(\tau) \right] &=& \mathrm{E}\left[ \sum_{i=1}^nK_h(t_{i-1}-\tau)\left(X_{k}(t_i) - X_{k}(t_{i-1})\right)\left(X_{l}(t_i)- X_{l}(t_{i-1})\right)\right]\\ \nonumber
&=&\sum_{i=1}^nK_h(t_{i-1}-\tau)\mathrm{E}\left[\left(X_{k}(t_i) - X_{k}(t_{i-1})\right)\left(X_{l}(t_i)- X_{l}(t_{i-1})\right)\right]\\
&=&\sum_{i=1}^nK_h(t_{i-1}-\tau)\int_{t_{i-1}}^{t_i}\theta_{kl}^{aa}(s)ds,
\end{eqnarray}
where the final equation is due to the results of Barndorff-Nielsen and Shephard \citeyearpar{barndorff2004a}:
\begin{equation}
\mathrm{E}\left[ \Delta X_k(t_{i-1})\Delta X_l(t_{i-1}) \right] = \int_{t_{i-1}}^{t_i}\theta_{kl}^{aa}(s)ds. 
\end{equation}
Next, we apply Lemma \ref{lemma1} and have
\begin{equation} \label{eq:integralSumKernel}
\sum_{i=1}^n K_h(t_{i-1}-\tau)\int_{t_{i-1}}^{t_i}\theta_{kl}^{aa}(s)ds = \int_0^TK_h(s-\tau)\theta_{kl}^{aa}(s)ds + o(\sqrt{\delta}).
\end{equation}
Thus
\begin{equation} \label{eq:expectationOfKernel}
\mathrm{E}\left[ \widehat{KCV}_{kl}(\tau) \right]=\int_0^TK_h(s-\tau)\theta_{kl}^{aa}(s)ds.
\end{equation}
In order to compute the covariance of $\widehat{KCV}_{kl}(\tau)$ in equation \eqref{KCVkl1} we use the following results from Barndorff-Nielsen and Shephard \citeyearpar{barndorff2004a}:
\begin{eqnarray}\label{eq:covofreal}
\text{Cov}\left\lbrace \left[ \Delta X_{k}(t_{i-1})\Delta X_{l}(t_{i-1} \right], \left[ \Delta X_{k'}(t_{i-1})\Delta X_{l'}(t_{i-1}\right]\right\rbrace = \\
\int_{t_{i-1}}^{t_i}\theta_{kk'}^{aa}(s)ds\int_{t_{i-1}}^{t_i}\theta_{kl'}^{cc}(s)ds +\int_{t_{i-1}}^{t_i}\theta_{kl'}^{aa}(s)ds\int_{t_{i-1}}^{t_i}\theta_{lk'}^{cc}(s)ds.
\end{eqnarray}
Now, using the definition of covariance and equations \eqref{eq:expectationOfKernel}, \eqref{eq:integralSumKernel} and \eqref{eq:covofreal} we have
\begin{eqnarray*}
&& \text{Cov}\left\lbrace  \widehat{KCV}_{kl}(\tau),\widehat{KCV}_{k'l'}(\tau) \right\rbrace \\
&=&\text{Cov} \left\lbrace [\sum_{i=1}^n K_h(t_{i-1}-\tau) \Delta X_{k}(t_{i-1})\Delta X_{l}(t_{i-1})], [\sum_{i=1}^n K_h(t_{i-1}-\tau) \Delta X_{k'}(t_{i-1})\Delta X_{l'}(t_{i-1})] \right\rbrace \\
&=&\mathrm{E} \lbrace [\sum_{i=1}^n K_h(t_{i-1}-\tau)\Delta X_{k}(t_{i-1})\Delta X_{l}(t_{i-1}) - \int_0^TK_h(s - \tau)\theta^{aa}_{kl}(s)ds] \\
&\times& [\sum_{i=1}^n K_h(t_{i-1}-\tau)\Delta X_{k'}(t_{i-1})\Delta X_{l'}(t_{i-1}) - \int_0^TK_h(s - \tau)\theta^{cc}_{k'l'}(s)ds] \rbrace \\
&=&\sum_{i=1}^nK_h^2(t_{i-1}-\tau)\lbrace \int_{t_{i-1}}^{t_{i}}\theta^{aa}_{kk'}(s)ds\int_{t_{i-1}}^{t_{i}}\theta^{cc}_{ll'}(s)ds + \int_{t_{i-1}}^{t_{i}}\theta^{aa}_{kl'}(s)ds\int_{t_{i-1}}^{t_{i}}\theta^{cc}_{lk'}(s)ds\rbrace.
\end{eqnarray*}
Apply Lemma \ref{lemma1} and invoke Riemann integration
\begin{eqnarray}
\delta^{-1}\sum_{i=1}^nK_h^2(t_{i-1}-\tau)\left\lbrace\int_{t_{i-1}}^{t_{i}}\theta^{aa}_{kk'}(s)ds\int_{t_{i-1}}^{t_{i}}\theta^{cc}_{ll'}(s)ds + \int_{t_{i-1}}^{t_{i}}\theta^{aa}_{kl'}(s)ds\int_{t_{i-1}}^{t_{i}}\theta^{cc}_{lk'}(s)ds\right\rbrace \nonumber\\
\to \int_0^T K_h^2(s-\tau)\Omega_{kl,k'l'}(s)ds, \nonumber
\end{eqnarray}
where $\Omega(s)$ is given in equation \eqref{Omega1}.
%Now applying Lemma 6 in Kristensen (2010) (equation \eqref{lemma6_2} in this note) and using result \eqref{covBarndResult}, we have 
%\begin{eqnarray} \label{CovofVariates}
%\frac{1}{\delta}\sum_{i=1}^nK_h^2(t_{i-1}-\tau)\text{Cov} [ \Delta X_{t_{i-1}}^k\Delta X_{t_{i-1}}^l], [ \Delta X_{t_{i-1}}^{k'}\Delta X_{t_{i-1}}^{l'}]=\int_0^TK_h^2(s-\tau)\Omega_{kl,k'l'}(s)ds,
%\end{eqnarray}	
\subsection{Asymptotic normality}\label{profTh1:AsympNorm}			
To prove the results of Theorem \ref{theorem1} in the case where the mean processes $\mu_k$ are identically $0$, we apply Cramer-Wold device, i.e. it suffices to show that for any real constants $a^{kl}$ we have, as $\delta \to 0$
\begin{equation}
\frac{1}{\sqrt{\delta}}(a^{kl}(\widehat{KCV}_{kl}(\tau) - \int_0^TK_h(s-\tau)\theta_{kl}^{aa}(s)ds) \overset{\mathcal{L}}{\to} \mathcal{N}(0,a^{kl}a^{k'l'}(\int_0^TK_h^2(s-\tau)\Omega_{kl,k'l'}(s)ds)).\label{asympDist}
\end{equation} 
We rewrite \eqref{asympDist} as
\begin{eqnarray}
\frac{1}{\sqrt{\delta}}\sum_{i=1}^n a^{kl}(K_h(t_{i-1}-\tau)\Delta X_k(t_{i-1}) \Delta X_l(t_{i-1}) - \int_{t_{i-1}}^{t_i}K_h(s-\tau)\theta^{aa}_{kl}(s)ds) \nonumber \\ 
\overset{\mathcal{L}}{\to} \mathcal{N}(0,a^{kl}a^{k'l'}\sum_{i=1}^n\int_{t_{i-1}}^{t_i}K_h^2(s-\tau)\Omega_{kl,k'l'}(s)ds).
\end{eqnarray}
Here we apply the Einstein summation convention also to the indices $k,l$. By the above calculations, 
\begin{eqnarray}
\text{Var} \left\lbrace \sum_{i=1}^n a^{kl}(K_h(t_{i-1}-\tau)\Delta X_k(t_{i-1}) \Delta X_l(t_{i-1}) - \int_{t_{i-1}}^{t_i}K_h(s-\tau)\theta^{aa}_{kl}(s)ds)\right\rbrace \nonumber \\ 
\to a^{kl}a^{k'l'}\sum_{i=1}^n\int_{t_{i-1}}^{t_i}K_h^2(s-\tau)\Omega_{kl,k'l'}(s)ds.
\end{eqnarray}
Now we apply the results of Linderberg-Feller Central Limit Theorem for triangular arrays of independent random variables $x_{n1}, \cdots, x_{nk_n} (n=1,2,\cdots, i=1,2,\cdots,k_n$ with $\ k_n \to \infty)$ and let $x_n=x_{n1}+ \cdots + x_{nk_n}$. We state the Corollary 3 from Barndorff-Nielsen and Shephard \citeyearpar{barndorff2004a} below.
\begin{cor}
Suppose that $\mathrm{E}[y_{ni}]=0$ for all $n$ and $i$ and there exists a nonnegative number $v$ that $\text{Var}[x_n] \to v$ for $n \to \infty$. Then 
\begin{equation}
y_n \overset{\mathcal{L}}{\to}\mathcal{N}(0,v)
\end{equation}
if and only if 
\begin{equation} \label{LFcondition}
\sum_{i=1}^{k_n}\mathrm{E}[y_{ni}^2\mathbf{1}_{(\psi,\infty)}(|y_{ni}|)] \to 0 \ \ \text{as} \ \ n\to \infty
\end{equation}
for an arbitrary $\psi >0$.
\end{cor}
A Lyapunov condition is sufficient for \eqref{LFcondition}, that is
\begin{equation}
\sum_{i=1}^{k_n}\mathrm{E}[|y_{ni}|^{2+\epsilon}] \to 0
\end{equation}
for some $\epsilon >0$. Now, let
\begin{equation}
y_{ni}=\frac{1}{\sqrt{\delta}}a^{kl}\lbrace K_h(t_{i-1}-\tau)\Delta X_k(t_{i-1}) \Delta X_l(t_{i-1}) - \int_{t_{i-1}}^{t_i}K_h(s-\tau)\theta^{aa}_{kl}(s)ds \rbrace,
\end{equation}
we find 
\begin{eqnarray}
y_{ni} &\overset{\mathcal{L}}{=}& \frac{1}{\sqrt{\delta}}a^{kl} \lbrace K_h(t_{i-1}-\tau)\sqrt{\int_{t_{i-1}}^{t_i}\theta^{aa}_{kk}(s)ds}\sqrt{\int_{t_{i-1}}^{t_i}\theta_{ll}^{cc}(s)ds}U_{kj}U_{lj} \nonumber\\
&&  \ \ \ \ \ \ \ \ \ \ \ \ \ \ \ \ \ \ \ \ \ \ \ \ \ \ \  \ \ \  \ \ \  \ \ \  \ \ \ -\int_{t_{i-1}}^{t_i}K_h(s-\tau)\theta^{aa}_{kl}(s)ds   \rbrace \nonumber \\
&\overset{\mathcal{L}}{=}&\frac{1}{\sqrt{\delta}}a^{kl} \lbrace K_h(t_{i-1}-\tau)\sqrt{\int_{t_{i-1}}^{t_i}\theta^{aa}_{kk}(s)ds}\sqrt{\int_{t_{i-1}}^{t_i}\theta^{cc}_{ll}(s)ds}U_{kj}U_{lj} \nonumber \\
&& \ \ \ \ \ \ \ \ \ \ \ \ \ \ \ \ \ \ \ \ \ \ \ \  \ \ \  \ \ \  \ \ \  \ \ \ -K_h(t_{i-1}-\tau)\int_{t_{i-1}}^{t_i}\theta^{aa}_{kl}(s)ds \rbrace \nonumber \\
&\overset{\mathcal{L}}{=}& \frac{1}{\sqrt{\delta}} \lbrace a^{kl}K_h(t_{i-1}-\tau)(\sqrt{\int_{t_{i-1}}^{t_i}\theta^{aa}_{kk}(s)ds}\sqrt{\int_{t_{i-1}}^{t_i}\theta^{cc}_{ll}(s)ds}(U_{kj}U_{lj}-\rho_{kl})\rbrace \nonumber \\
&\overset{\mathcal{L}}{=}&\sqrt{\delta}a^{kl}\lbrace K_h(t_{i-1}-\tau)\sqrt{\widehat{\Gamma}_{ki}\widehat{\Gamma}_{li}}(U_{kj}U_{lj}-\rho_{kl})\rbrace, \label{uniformbound}
\end{eqnarray}
where 
\begin{equation}
\tilde{\Gamma}_{ki}=\frac{1}{\delta}\int_{t_{i-1}}^{t_i}\theta^{aa}_{kk}(s)ds
\end{equation}
\begin{equation}
\end{equation}
and
\begin{equation}
\rho_{kl}= \frac{\int_{t_{i-1}}^{t_i}\theta^{aa}_{kl}(s)ds)}{\sqrt{\int_{t_{i-1}}^{t_i}\theta^{cc}_{kk}(s)ds}\sqrt{\int_{t_{i-1}}^{t_i}\theta^{dd}_{ll}(s)ds}}
\end{equation}
is the correlation coefficient between $U_{k}$ and $U_{l}$. By our Assumption on the process $\Sigma$, as $\delta$ varies the quantities $\Gamma$ are bounded away from $0$ and infinity, uniformly in $k$ and $j$. This implies that 
\begin{equation}
\mathrm{E}[|a^{kl}K_h(t_{i-1}-\tau)\sqrt{\widehat{\Gamma}_{ki}\widehat{\Gamma}_{li}}(U_{kj}U_{lj}-\rho_{kl}))|^{2+\epsilon}]
\end{equation}
is uniformly bounded above, and hence, by \eqref{uniformbound}, we have 
\begin{equation}
\sum_{i=1}^n\mathrm{E}[|y_{ni}|^{2+\epsilon}]\to 0
\end{equation}
as to be shown.
Next, we show that the effect of a nonzero drift term is negligible: 
\begin{eqnarray} \nonumber
\widehat{KCV}_{kl}(\tau) - \widehat{KCV}^*_{kl}(\tau)=\sum_{i=1}^nK_h(t_{i-1}-\tau)(\int_{t_{i-1}}^{t_i}\mu_k(s)ds)(\int_{t_{i-1}}^{t_i}\mu_l(s)ds) \\ \nonumber
+ \sum_{i=1}^nK_h(t_{i-1}-\tau)(\int_{t_{i-1}}^{t_i}\mu_k(s)ds)(\int_{t_{i-1}}^{t_i}\theta_l(s)dW(s)) \\
+ \sum_{i=1}^nK_h(t_{i-1}-\tau)(\int_{t_{i-1}}^{t_i}\mu_l(s)ds)(\int_{t_{i-1}}^{t_i}\theta_k(s)dW(s)). \label{eq:drifTerm}
\end{eqnarray}
By Lemma \ref{lemma1} the first term in equation \eqref{eq:drifTerm} is
\begin{eqnarray}
\sum_{i=1}^nK_h(t_{i-1}-\tau)(\int_{t_{i-1}}^{t_i}\mu_k(s)ds)(\int_{t_{i-1}}^{t_i}\mu_l(s)ds) \nonumber \\
=\delta\int_0^TK_h(s-\tau)\mu_k(s)\mu_l(s)ds + o(\delta).
\end{eqnarray}
The second term is
\begin{eqnarray*}
\sum_{i=1}^nK_h(t_{i-1}-\tau)(\int_{t_{i-1}}^{t_i}\mu_k(s)ds)(\int_{t_{i-1}}^{t_i}\theta_l(s)dW(s))\\ 
\sim\mathcal{N}\left( 0, \sum_{i=1}^nK^2_h(t_{i-1}-\tau)(\int_{t_{i-1}}^{t_i}\mu_k(s)ds)^2\int_{t_{i-1}}^{t_i}\theta_{ll}^{cc}(s)ds\right)
\end{eqnarray*}
and, similarly, the third term
\begin{eqnarray*}
\sum_{i=1}^nK_h(t_{i-1}-\tau)(\int_{t_{i-1}}^{t_i}\mu_l(s)ds)(\int_{t_{i-1}}^{t_i}\theta_k(s)dW(s))\\ 
\sim\mathcal{N}\left( 0, \sum_{i=1}^nK^2_h(t_{i-1}-\tau)\int_{t_{i-1}}^{t_i}\mu_l(s)ds)^2\int_{t_{i-1}}^{t_i}\theta_{kk}^{aa}(s)ds\right),
\end{eqnarray*}
where
\begin{eqnarray*}
\sum_{i=1}^nK^2_h(t_{i-1}-\tau)(\int_{t_{i-1}}^{t_i}\mu_k(s)ds)^2\int_{t_{i-1}}^{t_i}\theta_{ll}^{cc}(s)ds \\
\leq \delta \sup_{s}\theta_{ll}^{cc}(s)\times \sum_{i=1}^nK^2_h(t_{i-1}-\tau)(\int_{t_{i-1}}^{t_i}\mu_k(s)ds)^2\\
= \delta^2 \sup_{s}\theta_{ll}^{cc}(s)\times \left( \int_0^TK^2_h(s-\tau)\mu_k^2(s)ds +o(1) \right)
\end{eqnarray*}
and 
\begin{eqnarray*}
\sum_{i=1}^nK^2_h(t_{i-1}-\tau)(\int_{t_{i-1}}^{t_i}\mu_l(s)ds)^2\int_{t_{i-1}}^{t_i}\theta_{kk}^{aa}(s)ds \\
\leq \delta \sup_{s}\theta_{kk}^{aa}(s)\times \sum_{i=1}^nK^2_h(t_{i-1}-\tau)(\int_{t_{i-1}}^{t_i}\mu_l(s)ds)^2\\
= \delta^2 \sup_{s}\theta_{kk}^{aa}(s)\times \left( \int_0^TK^2_h(s-\tau)\mu_l^2(s)ds +o(1) \right).
\end{eqnarray*}

\section{Proof of Theorem \ref{theorem2}}\label{appendix:proofTheorem2}
The convergence results in the proof of Theorem \ref{theorem1} still holds when $h\to 0$. Now, we consider the shrinking bandwidth, $h\to 0$, and we derive the means and covariances of the varieties
\begin{eqnarray}
\hat{\Sigma}_{kl}(\tau)&=&\sum_{i=1}^n K_h(t_{i-1}-\tau)\Delta X_{k}(t_{i-1})\Delta X_{l}(t_{i-1})\\
&=&\sum_{i=1}^nK_h(t_{i-1}-\tau)\left(X_{k}(t_i) - X_{k}(t_{i-1})\right)\left(X_{l}(t_i)- X_{l}(t_{i-1})\right).
\label{eq:spotSigmakl1}
\end{eqnarray}
Following the proof of Theorem \ref{theorem1} and applying Lemma \ref{lemma2} we obtain:
\begin{equation*}
\sum_{i=1}^n K_h(t_{i-1} -t)\int_{t_i-1}^{t_i}\theta_{kl}^{aa}(s)d ds = \theta_{kl}^{aa} + h^{m+\gamma}\mathcal{K}(\tau,0)\int_\mathbb{R}K(z)z^{m+\gamma}dz + O\left(\frac{\delta}{h}\right) + O(h^{m+\gamma})
\end{equation*}
where $\mathcal{K}(\tau,0)$ denotes "Lipschitz coefficient" of $\theta_{kl}(s)$. Thus we have:
\begin{equation} \label{eq:expectationOfSpotKernel}
\mathrm{E}\left[ \hat{\Sigma}_{kl}(\tau) \right]= \Sigma_{kl}(\tau).
\end{equation}
For deriving the covariance of the varieties in \eqref{eq:spotSigmakl1} we use the following result from proof of theorem \ref{theorem1}:
\begin{eqnarray*}
&&\text{Cov}\left\lbrace  \hat{\Sigma}_{kl}(\tau),\hat{\Sigma}_{k'l'}(\tau) \right\rbrace = \\
&&\sum_{i=1}^nK_h^2(t_{i-1}-\tau)\lbrace \int_{t_{i-1}}^{t_{i}}\theta^{aa}_{kk'}(s)ds\int_{t_{i-1}}^{t_{i}}\theta^{cc}_{ll'}(s)ds + \int_{t_{i-1}}^{t_{i}}\theta^{aa}_{kl'}(s)ds\int_{t_{i-1}}^{t_{i}}\theta^{cc}_{lk'}(s)ds\rbrace.
\end{eqnarray*}
Now we using Lemma \ref{lemma2} and invoking Riemann integration we obtain:
\begin{eqnarray}
\delta^{-1}h\sum_{i=1}^nK_h^2(t_{i-1}-\tau)\left\lbrace\int_{t_{i-1}}^{t_{i}}\theta^{aa}_{kk'}(s)ds\int_{t_{i-1}}^{t_{i}}\theta^{cc}_{ll'}(s)ds + \int_{t_{i-1}}^{t_{i}}\theta^{aa}_{kl'}(s)ds\int_{t_{i-1}}^{t_{i}}\theta^{cc}_{lk'}(s)ds\right\rbrace \nonumber\\
\to \Omega_{kl,k'l'}(\tau)\int_{\mathbb{R}} K^2(z)dz, \nonumber
\end{eqnarray}
where $\Omega_{k,l,k',l'}(\tau)$ is defined in equation \eqref{Omega1}. One can easily show the asymptotic normality by following Section \ref{profTh1:AsympNorm} in the proof of Theorem \ref{theorem1} by applying Cramer-Wold device, i.e. to show that for any real constants $a^{kl}$ we have, as $\delta \to 0$ and $h \to 0$:
\begin{equation}
{\sqrt{\delta^{-1}h}}(a^{kl}(\widehat{\Sigma}_{kl}(\tau) - \Sigma_{kl}(\tau)) \overset{\mathcal{L}}{\to} \mathcal{N}(0,a^{kl}a^{k'l'}\Omega_{kl,k'l'}(\tau)\int_{\mathbb{R}} K^2(z)dz).\label{spotasympDist}
\end{equation} 
%, assumptions \eqref{assump2} and \eqref{assump3} enables us to apply Lemma 7 in \citeyearpar{Kristensen} and have 
%\begin{equation*}
%\sqrt{\delta^{-1}h}\left\lbrace \sum_{i=1}^n K_h(t_{i-1} -t)\int_{t_i-1}^{t_i}\Sigma_{kl}(s)d - \Sigma_kl\right\rbrace = O(\sqrt{\delta^{-1}h^{2(m+\gamma)+1}}) + O(1/\sqrt{\delta^{-1}h})
%\end{equation*}
%and
%\begin{equation*}
%\delta^{-1}h\sum_{i=1}^n(K_h(t_{i-1} -t)\int_{t_i-1}^{t_i}\Sigma_{kl}(s)ds)^2 -\Omega_{kl,k'l'}(t)\int_{\mathbb{R}}K^2(x)dx = O(\delta^\gamma ) + O(1/(\delta^{-1}h)).
%\end{equation*}

\section{Proof of Theorem \ref{theorem3}} \label{appendix:proofTheorem3}
Here we follow the notation in section \ref{notation} with $\Sigma_{kl}(t)$ denoting the $(k,l)$-th element of spot covariance matrix at time $t$ (see equation \eqref{eq:spotcovkl}). We first derive the asymptotic distribution of elements $(\Sigma_{kl})_{k,l=1,...d}$ of the covariance matrix by following Yu et al. \citeyearpar{YuFang} and then using Cram\'{e}r-Wold theorem prove multivariate convergence in distribution using univariate results.

Let $\widehat{TCV}_{kl}$ denote the $(k,l)$-th component of the estimator and $X^*$ denote the diffusion part of $X$. So, we have 
\begin{eqnarray} \label{eq:jumpproof}
&&\sqrt{n}\frac{\widehat{TCV}_{kl}-\int_0^TK_h(s-t)\Sigma_{kl}(s)ds}{\sqrt{\int_0^TK^2_h(s-t)\Omega_{kl,k'l'}(s)ds}} \nonumber \\
&=& \sqrt{n}\frac{\sum_{i=1}^nK_h(t_{i-1}-t)\Delta_{i-1}X^*_{kl}\Delta_{i-1}X^*_{k'l'}\mathbbm{1}_{\{ \Delta_{i-1}N=0 \} } - \int_0^TK_h(s-t)\Sigma_{kl}(s)ds}{\sqrt{\int_0^TK^2_h(s-t)\Omega_{kl,k'l'}(s)ds}} \nonumber \\
&=& \sqrt{n}\frac{\sum_{i=1}^nK_h(t_{i-1}-t)\Delta_{i-1}X^*_{kl}\Delta_{i-1}X^*_{k'l'} - \int_0^TK_h(s-t)\Sigma_{kl}(s)ds}{\sqrt{\int_0^TK^2_h(s-t)\Omega_{kl,k'l'}(s)ds}} \nonumber \\
&-& \sqrt{n} \frac{\sum_{i=1}^nK_h(t_{i-1}-t)\Delta_{i-1}X^*_{kl}\Delta_{i-1}X^*_{k'l'}\mathbbm{1}_{\{ \Delta_{i-1}N\neq 0 \} }}{\sqrt{\int_0^TK^2_h(s-t)\Omega_{kl,k'l'}(s)ds}}.
\end{eqnarray}
The first term in equation \eqref{eq:jumpproof} for the fixed $h$, as $\delta \to 0$ is
\begin{equation}
\sqrt{n}\frac{\sum_{i=1}^nK_h(t_{i-1}-t)\Delta_{i-1}X^*_{kl}\Delta_{i-1}X^*_{k'l'} - \int_0^TK_h(s-t)\Sigma_{kl}(s)ds}{\sqrt{\int_0^TK^2_h(s-t)\Omega_{kl,k'l'}(s)ds}} \overset{\mathcal{L}}{\to} N(0,1).
\end{equation}
Now, the assumption \ref{assump4} states that the kernel $K$ is bounded $|K_h(t_{i-1}-t_i)|\leq \Lambda/h$ for some constant $\Lambda$. the number of jumps occurring over the interval $[0,T]$ is finite. Then the second term in equation \eqref{eq:jumpproof}
\begin{eqnarray}
\sqrt{n}\frac{ \left| \sum_{i=1}^nK_h(t_{i-1}-t)\Delta_{i-1}X^*_{kl}\Delta_{i-1}X^*_{k'l'}\mathbbm{1}_{\{ \Delta_{i-1}N\neq 0 \} } \right|}{\sqrt{\int_0^TK^2_h(s-t)\Omega_{kl,k'l'}(s)ds}} &\leq \\
\sqrt{n}\frac{\frac{\Lambda}{h} \sum_{i=1}^n \Delta_{i-1}X^*_{kl}\Delta_{i-1}X^*_{k'l'}\mathbbm{1}_{\{ \Delta_{i-1}N\neq 0 \} } }{\sqrt{\int_0^TK^2_h(s-t)\Omega_{kl,k'l'}(s)ds}} \nonumber &\leq \\
\sqrt{n}\frac{N_T \times \frac{\Lambda}{h} \times \sup \int_{t_{i-1}}^{t_i}\theta_{kl}(s)dW_s \int_{t_{i-1}}^{t_i}\theta_{k'l'}(s)dB_s}{\sqrt{\int_0^TK^2_h(s-t)\Omega_{kl,k'l'}(s)ds}} \nonumber.
\end{eqnarray}
Here the integral $\int_0^t\theta_{kl}(s)dW_s$, $\int_{0}^{t}\theta_{k'l'}(s)dB_s$ are time changed Brownian motion and by the Levy law of the modulus of continuity of Brownian motion's path Karatzas and Shreve \citeyearpar{KaratzasShreve}, for small $s$ we have
\begin{equation}
\sup_{i\in{1,...,n}}\frac{\left| \int_{t_{i-1}}^{t_i}\theta_{kl}(s)dW_s \right|}{\sqrt{2\delta \log \frac{1}{\delta}}} \leq \sqrt{M}, \ \ \ \ \sup_{i\in{1,...,n}}\frac{\left| \int_{t_{i-1}}^{t_i}\theta_{k'l'}(s)dB_s \right|}{\sqrt{2\delta \log \delta^{-1}}} \leq \sqrt{L},
\end{equation}
where $M, L$ are a non-negative constants. Therefor the last term in equation \eqref{eq:jumpproof} is
\begin{eqnarray}
\sqrt{n}\frac{ \left| \sum_{i=1}^nK_h(t_{i-1}-t)\Delta_{i-1}X^*_{kl}\Delta_{i-1}X^*_{k'l'}\mathbbm{1}_{\{ \Delta_{i-1}N\neq 0 \} } \right|}{\sqrt{\int_0^TK^2_h(s-t)\Omega_{kl,k'l'}(s)ds}} \leq \\
\sqrt{n}\frac{N_T \times \frac{\Lambda}{h} \times \sqrt{M} \times \sqrt{L} \times 2 \delta \log \delta^{-1} }{\sqrt{\int_0^TK^2_h(s-t)\Omega_{kl,k'l'}(s)ds}} \overset{P}{\to} 0. \nonumber
\end{eqnarray}
To prove multivariate convergence, given that we have asymptotic distribution of the elements of the covariance matrix, we employ Cram\'{e}r-Wold device:
\begin{lem}
For any real $a\in \mathbb{R}^{d\times d}$, as $\delta \to 0$
\begin{equation}
a^T\left( \frac{1}{\sqrt{\delta}}(\widehat{TCV}(t) - \int_0^TK_h(s-t)\Sigma(s)ds ) \right) \overset{\mathcal{L}}{\to} a^TN\left(0, \int_0^TK_h^2(s-t)\Omega(s)ds\right) a.
\end{equation}
\end{lem}
\begin{proof}
This follows from univariate case, since $a^{kl}\left( \frac{1}{\sqrt{\delta}}(\widehat{TCV}_{kl}(t) - \int_0^TK_h(s-t)\Sigma_{kl}(s)ds)\right)$ for $k,l=1,...,d$ are independent and identically distributed with mean $0$ and variance $a^{kl}a^{k'l'}\int_0^TK_h^2(s-t)\Omega_{kl,k'l'}(s)ds$.
\end{proof}

\section{Lemmas} \label{proofLemma}
We rewrite the Lemma 6 and Lemma 7 from Kristensen \citeyearpar{Kristensen} in terms of the components of covariance matrix.
\begin{lem} \label{lemma1}
Under Assumption \ref{assump2} and Assumption \ref{assump3}, we have for every  $k,l=1, \cdots, d$
\begin{subequations}
\begin{empheq}{align}
&(i)\ \  \sum_{i=1}^n K_h(t_{i-1}-\tau)\int_{t_{i-1}}^{t_i} \Sigma_{kl}(s)ds = \int_0^TK_h(s-\tau)\Sigma_{kl}(s)ds + o(\delta)\bar{K}_{1},  \nonumber \\
&(ii)\ \  \delta^{-1}\sum K_h^2(t_{i-1}-\tau)\left( \int_{t_{i-1}}^{t_i} \Sigma_{kl}(s)ds \right)^2 = \int_0^T K_h^2(s-\tau)\Sigma_{kl}^2(s)ds + o(1)\times \bar{K}_0 \nonumber \\
& \ \ \ \ \ \ \ \ \ \ \ \ \ \ \ \ \ \ \ \ \ \ \ \ \ \ \ \ \ \ \ \ \ \ \ \ \ \ \ \ \ \ \ \ \ \ \ \ \ \ \ \ \ \ \ \ \ \ + O(\delta)\times \bar{K}_{1} \nonumber
\end{empheq}
\end{subequations}
uniformly over $\tau \in [0,T]$, as $\delta \to 0$.
\end{lem}
\begin{proof}
See Kristensen \citeyearpar{Kristensen} Lemma 6. 
\end{proof}

\begin{lem} \label{lemma2}
Under Assumption \ref{assump4} $a$, $b$ and Assumption \ref{assump5}, uniformly over $\tau \in [a, T-a]$, as $\delta, h, a/h \to 0$ we have:
\begin{subequations}
\begin{empheq}{align} \nonumber
&(i)\ \  \sum_{i=1}^n K_h(t_{i-1}-\tau)\int_{t_{i-1}}^{t_i} \Sigma_{kl}(s)ds = \Sigma_{kl}(\tau) + h^{m+\gamma}\mathcal{K}(\tau,0)\int_{\mathbb{R}}K(z)z^{m+\gamma}dz +  \\ \nonumber
&  \ \ \ \ \ \ \ \ \ \ \ \ \ \ \ \ \ \ \ \ \ \ \ \ \ \ \ \ \ \ \ \ \ \ \ \ \ \ \ \ \ \ \ \ \ \ \ O\left(\frac{\delta}{h}\right) + o(h^{m+\gamma}), \\ \nonumber
&(ii)\ \  \sum K_h^2(t_{i-1}-\tau)\left( \int_{t_{i-1}}^{t_i} \Sigma_{kl}(s)ds \right)^2 = \frac{\delta}{h}\Sigma_{kl}^2(\tau)\int_{\mathbb{R}} K^2(z)dz + O\left(\frac{\delta^{1+\gamma}}{h}\right) +O\left(\frac{\delta^2}{h^2}\right).  
\end{empheq}
\end{subequations}
\end{lem}
\begin{proof}
See Kristensen \citeyearpar{Kristensen} Lemma 7. 
\end{proof}

\end{appendices}

\end{document}